\documentclass[11pt]{article}

\usepackage[utf8]{inputenc}
\usepackage[T1]{fontenc}
\usepackage{amsmath,amssymb,verbatim}
\usepackage{url,enumerate,graphicx,graphics}
\usepackage[boxruled,lined,linesnumbered]{algorithm2e}

\newcommand{\ignore}[1]{{}}

\newenvironment{proof}{\par\noindent{\bf Proof:}}{\mbox{}\hfill$\Box$\\}

\newtheorem{definition}{Definition}[section]

\newtheorem{theorem}{Theorem}[section]

\newtheorem{lemma}{Lemma}[section]

\newcommand{\concat}{$\normalfont::$}

\newcommand{\inedges}{\textsc{In-Edges}}

\newcommand{\kftrs}{\textsc{$k$-FTRS}}

\newcommand{\Out}{\textsc{\scriptsize{out}}}
\newcommand{\In}{\textsc{\scriptsize{in}}}

\newcommand{\etal}{et al.}

\begin{document}

\title{An efficient  strongly connected components algorithm\\ in the fault tolerant model}

\author{Surender Baswana\\
   Department of CSE,\\
   I.I.T. Kanpur, India\\
   {\small\texttt{sbaswana@cse.iitk.ac.in}}
 \and
   Keerti Choudhary\\
   Department of CSE,\\
   I.I.T. Kanpur, India\\
   {\small\texttt{keerti@cse.iitk.ac.in}}
 \and
   Liam Roditty\\
   Department of Comp. Sc.\\
   Bar Ilan University, Israel.\\
   {\small\texttt{liam.roditty@biu.ac.il}}
}

\date{}
\maketitle

\begin{abstract}

In this paper we study the problem of maintaining the strongly connected components
of a graph in the presence of failures.
In particular, we show that given a directed graph $G=(V,E)$ with $n=|V|$ and $m=|E|$,
and an integer value $k\geq 1$, there is an algorithm that computes in $O(2^{k}n\log^2 n)$ time 
for any set $F$ of size at most $k$  the strongly connected components of the graph $G\setminus F$.
The running time of our algorithm is almost optimal since the time for outputting the SCCs of $G\setminus F$ 
is at least $\Omega(n)$. The algorithm uses a data structure that is computed in a preprocessing 
phase in polynomial time and is of size $O(2^{k} n^2)$.

Our result is obtained using a new observation on the relation between strongly connected components 
(SCCs) and reachability. More specifically, one of the main building blocks in our result is a restricted variant 
of the problem in which we only compute strongly connected components that intersect a certain path. 
Restricting our attention to a path allows us to implicitly compute reachability between the path vertices and
the rest of the graph in time that depends logarithmically  rather than linearly  in  the  size of the path. 
This new observation alone, however, is not enough, since we need to find an efficient way to represent 
the strongly connected components using paths. For this purpose we use a mixture of old and classical 
techniques such as the heavy path decomposition of Sleator and Tarjan~\cite{ST:83} and the classical 
Depth-First-Search algorithm. Although, these are by now standard techniques, we are not aware of any 
usage of them in the context of dynamic maintenance of SCCs. Therefore, we expect that our new insights 
and mixture of new and old techniques will be of independent interest.

\end{abstract}


\section{Introduction}

Computing the strongly connected components (SCCs) of a directed graph $G=(V,E)$, where $n=|V|$ and $m=|E|$,
is one of the most fundamental problems in computer science. There are several classical algorithms for computing
the SCCs in $O(m+n)$ time that are taught in any standard undergrad algorithms course~\cite{Cormen}.

In this paper we study the following natural variant of the problem in dynamic graphs. What is the fastest algorithm 
to compute the SCCs of $G\setminus F$, where $F$ is any  set of edges or vertices. The algorithm can use a 
polynomial size  data structure computed in polynomial time for $G$ during a preprocessing phase.

The main result of this paper is:
\begin{theorem}\label{T-main}
There is an algorithm that computes the SCCs of $G\setminus F$, for any set $F$ of $k$ edges or vertices,
in $O(2^{k}n\log^2 n)$ time. The algorithm uses a data structure of size $O(2^{k} n^2)$ computed in
$O(2^kn^2m)$ time for $G$ during a preprocessing phase.
\end{theorem}

Since the time for outputting the SCCs of $G\setminus F$ is at least $\Omega(n)$,
the running time of our algorithm is optimal (up to a polylogarithmic factor) for any fixed value of $k$.

This dynamic model is usually called the fault tolerant model  and its most important parameter is the time 
that it takes to compute the output in the presence of faults. It is an important theoretical model as it can be
viewed as a restriction of the deletion only (decremental) model in which  edges (or vertices) are deleted 
one after another and queries are answered between deletions. The fault tolerant  model is especially useful 
in cases where the worst case update time in the more general decremental model is high.

There is wide literature on the problem of decremental SCCs. Recently, in a major breakthrough, Henzinger, 
Krinninger and Nanongkai \cite{HenzingerKN14} presented a randomized algorithm with $O(mn^{0.9+o(1)})$ 
total update time and broke the barrier of $\Omega(mn)$ for the problem. Even more recently, 
Chechik et al.~\cite{ChETAL} obtained an improved  total  running time of $O(m\sqrt{n \log n})$.

However, these algorithms and in fact all the previous algorithms have an $\Omega(m)$ worst case update time 
for a single edge deletion. This is not a coincidence. Recent developments in conditional lower bounds by Abboud 
and V. Williams~\cite{AbboudW14} and  by Henzinger, Krinninger, Nanongkai and Saranurak~\cite{HenzingerKNS15}  
showed that unless a major breakthrough happens, the worst case update time of a single operation in any algorithm 
for decremental SCCs is $\Omega(m)$. Therefore, in order to obtain further theoretical understanding on the problem 
of decremental SCCs, and in particular on the worst case update time it is only natural to focus on the restricted 
dynamic model of fault tolerant.

In the recent decade several different researchers used the fault tolerant model to study the worst case update 
time per operation  for dynamic  connectivity in undirected graphs. P\v{a}tra\c{s}cu and Thorup~\cite{PaTh07} 
presented connectivity algorithms that support edge deletions in this model. Their result was improved by the 
recent polylogarithmic  worst case update time algorithm of Kapron, King and Mountjoy~\cite{KapronKM13}. 
Duan and Pettie\cite{DuPe10, DuPe17} used this model to obtain connectivity algorithms that support vertex 
deletions.

In directed graphs, very recently, Georgiadis, Italiano and Parotsidis~\cite{GeItPa17} considered the 
problem of SCCs but only for a single edge or a single vertex failure, that is $|F|=1$. They showed that it is 
possible to compute the SCCs of $G\setminus \{e\}$ for any $e\in E$ (or of $G\setminus \{v\}$ for any $v \in V$) 
in $O(n)$ time using a data structure of size $O(n)$ that was computed for $G$ in a preprocessing phase in 
$O(m+n)$ time. Our result is the first generalized result for any constant size $F$. This comes with the price 
of an extra $O(\log^2 n)$ factor in the running time, a slower preprocessing time and a larger data structure.
In \cite{GeItPa17}, Georgiadis, Italiano and Parotsidis also considered the problem of answering strong connectivity 
queries after one failure. They show construction of an $O(n)$ size oracle that can answer in constant time 
whether any two given vertices of the graph are strongly connected after failure of a single edge or a single vertex.

In a recent result~\cite{OUR-STOC} we considered the problem of finding a sparse subgraph that
preserves single source reachability. More specifically, given a directed graph $G=(V,E)$ and a vertex $s\in V$,
a subgraph $H$ of $G$ is said to be a $k$-Fault Tolerant Reachability Subgraph ($\kftrs$) for $G$  if for any set
$F$ of at most $k$ edges (or vertices), a vertex $v\in V$ is reachable from $s$ in $G\setminus F$ if and only if
$v$ is reachable from $s$ in $H\setminus F$. In~\cite{OUR-STOC} we proved that there exists a $\kftrs$
for $s$ with at most $2^k n$ edges.

Using the $k$-FTRS structure, it is relatively straightforward to obtain a data structure that, for any pair of vertices 
$u,v \in V$ and any set $F$ of size $k$, answers in $O(2^{k} n)$ time queries of the form:
$$\text{``Are $u$ and $v$ in the same SCC of $G\setminus F$?''}$$

The data structure consists of a $\kftrs$ for every $v\in V$. It is easy to see that $u$ and $v$ are in the same 
SCC of $G\setminus F$ if and only if $v$ is reachable from $u$  in $\kftrs(u)\setminus F$ and $u$ is reachable 
from $v$  in $\kftrs(v)\setminus F$. So the query can be answered by checking, using graph traversals, 
whether $v$ is reachable from $u$  in $\kftrs(u)\setminus F$ and whether $u$ is reachable from $v$ in 
$\kftrs(v)\setminus F$. The cost of these two graph traversals is $O(2^k n )$. The size of the data structure 
is $O(2^k n^2)$.

This problem, however, is much easier since the vertices in the query reveal which two $\kftrs$ we  need to  
scan. In the challenge that we address in this paper \emph{all} the SCCs of $G\setminus F$, for an arbitrary 
set $F$, have to be  computed.  However, using the same data structure as before, it is not really clear 
a-priori which of the $\kftrs$ we need to scan.

We note that our algorithm uses the $\kftrs$ which seems to be an essential  tool but is far from being a 
sufficient one and more involved ideas are required. As an example to such a relation between a new result 
and an old tool one can take the deterministic algorithm of {\L}\k{a}cki~\cite{Lacki11} for decremental SCCs 
in which the classical algorithm of Italiano~\cite{Italiano88} for decremental reachability trees in directed 
acyclic graphs is used. The main  contribution of {\L}\k{a}cki~\cite{Lacki11} is a new graph decomposition 
that made it possible to use Italiano's algorithm~\cite{Italiano88} efficiently.

\subsection{An overview of our result}
\label{subsection:overview}
We obtain our $O(2^k n \log^2 n)$-time algorithm using several new ideas. 
One of the main building blocks is surprisingly the following restricted variant of the problem.

\vspace{3mm}
{\em Given any set $F$ of $k$ failed edges and any path $P$ which is intact in $G\setminus F$,
output all the SCCs of $G\setminus F$ that intersect with $P$ (i.e. contain at least one vertex of $P$).}
\vspace{3mm}

To solve this restricted version, we implicitly solve the problem of reachability from $x$ (and to $x$) in 
$G\setminus F$, for each $x\in P$. Though it is trivial to do so in time $O(2^kn|P|)$ using $\kftrs$ of 
each vertex on $P$, our goal is to preform this computation in $O(2^k n \log n)$ time, that is, in running 
time that is \emph{independent} of the length of $P$ (up to a logarithmic factor). For this we use a careful 
insight into the structure of reachability between $P$ and $V$. Specifically, if $v\in V$ is reachable from 
$x\in P$, then $v$ is also reachable from any predecessor of $x$ on $P$, and if $v$ is not reachable 
from $x$, then it cannot be reachable from any successor of $x$ as well. Let $w$ be any vertex on $P$, 
and let $A$ be the set of vertices reachable from $w$ in $G\setminus F$. Then we can split $P$ at $w$ 
to obtain two paths: $P_1$ and $P_2$. We already know that all vertices in $P_1$ have a path to $A$, 
so for $P_1$ we only need to focus on set $V\setminus A$. Also the set of vertices reachable from any 
vertex on $P_2$ must be a subset of $A$, so for $P_2$ we only need to focus on set $A$. This suggests 
a divide-and-conquer approach which along with some more insight into the structure of $\kftrs$
helps us to design an efficient algorithm for computing all the SCCs that intersect $P$.

In order to use the above result to compute all the SCCs of $G\setminus F$, we need a clever partitioning 
of $G$ into a set of vertex disjoint paths. A Depth-First-Search (DFS) tree plays a crucial role here as follows.
Let $P$ be any path from root to a leaf node in a DFS tree $T$. If we compute the SCCs intersecting $P$
and remove them, then the remaining SCCs must be contained in subtrees hanging from path $P$. So to 
compute the remaining SCCs we do not need to work on the entire graph. Instead, we need to work on 
each subtree. In order to pursue this approach efficiently, we need to select path $P$ in such a manner 
that the subtrees hanging from $P$ are of small size. The heavy path decomposition of Sleator and 
Tarjan~\cite{ST:83} helps to achieve this objective.%
\footnote{We note that the heavy path decomposition was also used in the fault tolerant model in STACS'10 paper of
\cite{BaswanaK10-stacs}, but in a completely different way and for a different problem.}

Our algorithm and data structure can be extended to support insertions as well. More specifically,
we can report the SCCs of a graph that is updated by insertions and deletions of $k$ edges in the 
same running time.

\subsection{Related work}

The problem of maintaining the SCCs of a graph was studied in the decremental model. In this model 
the goal is to maintain the SCCs of a graph whose edges are being deleted by an adversary. The main 
parameters in this model are the worst case update time per an  edge deletion and the total update from 
the first edge deletion until the last. Frigioni~\etal \cite{FMNZ01} presented an algorithm that has an
\emph{expected} total update time of $O(mn)$ if all the deleted edges are chosen at random. Roditty 
and Zwick~\cite{RoZw08} presented a Las-Vegas algorithm with an \emph{expected} total update time 
of $O(mn)$ and \emph{expected} worst case update time per a single edge deletion of $O(m)$.
{\L}\k{a}cki~\cite{Lacki11} presented a deterministic algorithm with a total update time of $O(mn)$, and 
thus solved the open problem posed by Roditty and Zwick in~\cite{RoZw08}. However, the worst case 
update time per a single edge deletion of his algorithm is $O(mn)$. Roditty~\cite{Roditty13} 
improved the worst case update time of a single edge deletion to $O(m\log n)$. Recently, in a major 
breakthrough, Henzinger, Krinninger and Nanongkai \cite{HenzingerKN14} presented a randomized 
algorithm with $O(mn^{0.9+o(1)})$ total update time. Very recently, Chechik et al.~\cite{ChETAL} obtained a 
total  update time of $O(m\sqrt{n \log n})$. Note that all the previous works on decremental SCC are with 
$\Omega(m)$ worst case update time. Whereas, our result directly implies $O(n \log^2 n)$ worst case 
update time as long as the total deletion length is constant.

Most of the previous work in the fault tolerant model is on variants of the shortest path problem. Demetrescu, 
Thorup, Chowdhury and Ramachandran~\cite{DT:08} designed an $O(n^2 \log n)$ size data structure that can 
report the distance from $u$ to $v$ avoiding $x$ for any $u,v,x\in V$ in $O(1)$ time. Bernstein and 
Karger~\cite{BK:9} improved the preprocessing time of~\cite{DT:08} to $O(mn ~{\mbox{polylog}}~ n)$. Duan 
and Pettie~\cite{DP:09} designed such a data structure for two vertex faults of size $O(n^2 \log n)$. Weimann 
and Yuster~\cite{WY:13} considered the question of optimizing the preprocessing time using Fast Matrix 
Multiplication (FMM) for graphs with integer weights from the range $[-M,M]$. Grandoni and Vassilevska 
Williams \cite{GW:12} improved the result of \cite{WY:13} based on a novel algorithm for computing all the
replacement paths from a given source vertex in the same running time as solving APSP in directed graphs.

For the problem of single source shortest paths Parter and Peleg~\cite{ParterP:13} showed that there is a 
subgraph with $O(n^{3/2})$ edges that supports one fault. They also showed a matching lower bound. 
Recently, Parter~\cite{Parter15} extended this result to two faults with $O(n^{5/3})$ edges for undirected 
graphs. She also showed a lower bound of $\Omega(n^{5/3})$.

Baswana and Khanna \cite{BaswanaK10-stacs} showed that there is a subgraph with $O(n \log n)$ edges that 
preserves the distances from $s$ up to a multiplicative stretch of $3$ upon failure of any single vertex. 
For the case of edge failures, sparse fault tolerant subgraphs exist for general $k$. 
Bil{\`{o}} et al.~\cite{Bilo16} showed that we can compute a subgraph with $O(kn)$ edges that preserves 
distances from $s$ up to a multiplicative stretch of $(2k +1)$ upon failure of any $k$ edges. They also 
showed that we can compute a data structure of $O(kn \log^2 n)$ size that is able to report the
$(2k +1)$-stretched distance from $s$ in $O(k^2 \log^2 n)$ time.

The questions of finding graph spanners, approximate distance oracles and
compact routing schemes in the fault tolerant model were studied in \cite{DK:11,CR:12,C:13,ChechikCFK17}.

\subsection{Organization of the paper}
We describe notations, terminologies, some basic properties of DFS, heavy-path decomposition, and $\kftrs$ 
in Section 2. In Section 3, we describe the fault tolerant algorithm for computing the strongly connected 
components intersecting any path. We present our main algorithm for handling $k$ failures in Section 4.
In Section 5, we show how to extend our algorithm and data structure to also support insertions.


\section{Preliminaries}

Let $G=(V,E)$ denote the input directed graph on $n=|V|$ vertices and $m=|E|$ edges.
We assume that $G$ is strongly connected, since if it is not the case, then we may apply our result
to each strongly connected component of $G$.
We first introduce some notations that will be used throughout the paper.

\begin{itemize}
\item $T$:~ A DFS tree of $G$.
\item $T(v)$:~ The subtree of $T$ rooted at a vertex $v$.
\item $Path(a,b)$:~ The tree path from $a$ to $b$ in $T$. Here $a$ is assumed to be an ancestor of $b$.
\item $depth(Path(a,b))$:~ The depth of vertex $a$ in $T$.
\item $G^R$:~ The graph obtained by reversing all the edges in graph $G$.
\item $H(A)$:~ The subgraph of a graph $H$ induced by the vertices of subset $A$.
\item $H\setminus F$:~ The graph obtained by deleting the edges in set $F$ from graph $H$.
\item $\inedges(v,H)$:~ The set of all incoming edges to $v$ in graph $H$.
\item $P[a,b]$:~ The subpath of path $P$ from vertex $a$ to vertex $b$, assuming $a$ and $b$ are in $P$ and $a$ precedes $b$.
\item $P\concat Q$~:~ The path formed by concatenating paths $P$ and $Q$ in $G$.
Here it is assumed that the last vertex of $P$ is the same as the first vertex of $Q$.
\end{itemize}

Our algorithm for computing SCCs in a fault tolerant environment crucially uses the concept of a
$k$-fault tolerant reachability subgraph ($k$-FTRS) which is a sparse subgraph that preserves reachability
from a given source vertex even after the failure of at most $k$ edges in $G$. A $k$-FTRS is formally defined as follows.

\begin{definition}[$\kftrs$]
Let $s\in V$ be any designated source. A subgraph $H$ of $G$ is said to be a $k$-Fault
Tolerant Reachability Subgraph ($\kftrs$) of $G$ with respect to $s$ if for any subset $F\subseteq E$
of $k$ edges, a vertex $v\in V$ is reachable from $s$ in $G\setminus F$ if and only if $v$ is reachable
from $s$ in $H\setminus F$.
\label{definition:k-FTRS}
\end{definition}

In~\cite{OUR-STOC}, we present the following result for the construction of a $\kftrs$ for
any $k\ge 1$.

\begin{theorem}[\cite{OUR-STOC}]
There exists an $O(2^kmn)$ time algorithm that for any given integer $k\geq 1$, and any given directed
graph $G$ on $n$ vertices, $m$ edges and a designated source vertex $s$, computes a $\kftrs$ for $G$ with
at most $2^k n$ edges. Moreover, the in-degree of each vertex in this $\kftrs$ is bounded by $2^k$.
\label{theorem:kftrs}
\end{theorem}

Our algorithm will require the knowledge of the vertices reachable from a vertex $v$ as well as the  vertices that can reach $v$.
So we define a $\kftrs$ of both the graphs - $G$ and $G^R$ with respect to any source vertex $v$ as follows.
\begin{itemize}
\item ${\cal G}^{}(v)$:~ The $\kftrs$ of graph $G$ with $v$ as source obtained by Theorem \ref{theorem:kftrs}.
\item ${\cal G}^{R}(v)$:~ The $\kftrs$ of graph $G^R$ with $v$ as source obtained by Theorem \ref{theorem:kftrs}.
\end{itemize}

The following lemma states that the subgraph of a $\kftrs$ induced by $A\subset V$ can serve as a $\kftrs$
for the subgraph $G(A)$ given that $A$ satisfies certain properties.

\begin{lemma}
Let $s$ be any designated source and $H$ be a $\kftrs$ of $G$ with respect to $s$.
Let $A$ be a subset of $V$ containing $s$ such that every path from $s$ to any vertex in $A$ is contained in $G(A)$.
Then $H(A)$ is a $\kftrs$ of $G(A)$ with respect to $s$.
\label{lemma:kftrs-set}
\end{lemma}
\begin{proof}
Let $F$ be any set of at most $k$ failing edges, and $v$ be any vertex reachable from $s$ in
$G(A)\setminus F$. Since $v$ is reachable from $s$ in $G\setminus F$ and $H$ is a $\kftrs$ of $G$,
so $v$ must be reachable from $s$ in $H \setminus F$ as well. Let $P$ be any path from $s$ to $v$
in $H\setminus F$. Then (i) all edges of $P$ are present in $H$ and (ii) none of the edges of $F$
appear on $P$. Since it is already given that every path from $s$ to any vertex in $A$ is contained
in $G(A)$, therefore, $P$ must be present in $G(A)$. So every vertex of $P$ belongs to $A$. This fact
combined with the inferences (i) and (ii) imply that $P$ must be present in $H(A)\setminus F$.
Hence $H(A)$ is $\kftrs$ of $G(A)$ with respect to $s$.
\end{proof}

The next lemma is an adaptation of Lemma 10 from Tarjan's 
classical paper on Depth First Search~\cite{Tarjan72} to our needs.

\begin{lemma}
Let $T$ be a DFS tree of $G$. Let $a,b\in V$ be two vertices without any ancestor-descendant relationship in $T$, and
assume that $a$ is visited before $b$ in the DFS traversal of $G$ corresponding to tree $T$.
Every path from $a$ to $b$ in $G$ must pass through a common ancestor of $a,b$ in $T$.
\label{lemma:DFS-property}
\end{lemma}
\begin{proof}
Let us assume on the contrary that there exists a path $P$ from $a$ to $b$ in $G$ that does not
pass through any common ancestor of $a$, $b$ in $T$. Let $z$ be the LCA of $a,b$ in $T$, and $w$ be
the child of $z$ lying on $Path(z,a)$ in $T$. See Figure \ref{figure:dfs}.
Let $A$ be the set of vertices which are either visited before $w$ in $T$ or lie in the subtree $T(w)$,
and $B$ be the set of vertices visited after $w$ in $T$.
Thus $a$ belongs to set $A$, and $b$ belongs to set $B$.
Let $x$ be the last vertex in $P$ that lies in set $A$, and $y$ be the successor of $x$ on path $P$.
Since none of vertices of $P$ is a common ancestor of $a$ and $b$, therefore, the edge $(x,y)$ must belong to set $A\times B$.
So the following relationship must hold true-
$\textsc{Finish-Time}(x)\leq\textsc{Finish-Time}(w)<\textsc{Visit-Time}(y)$.
But such a relationship is not possible since all the out-neighbors of $x$ must be visited before the
DFS traversal finishes for vertex $x$. Hence we get a contradiction.
\end{proof}

\begin{figure}[!ht]
\centering
\includegraphics[scale=0.5, trim=20mm 78mm 40mm 58mm,clip]{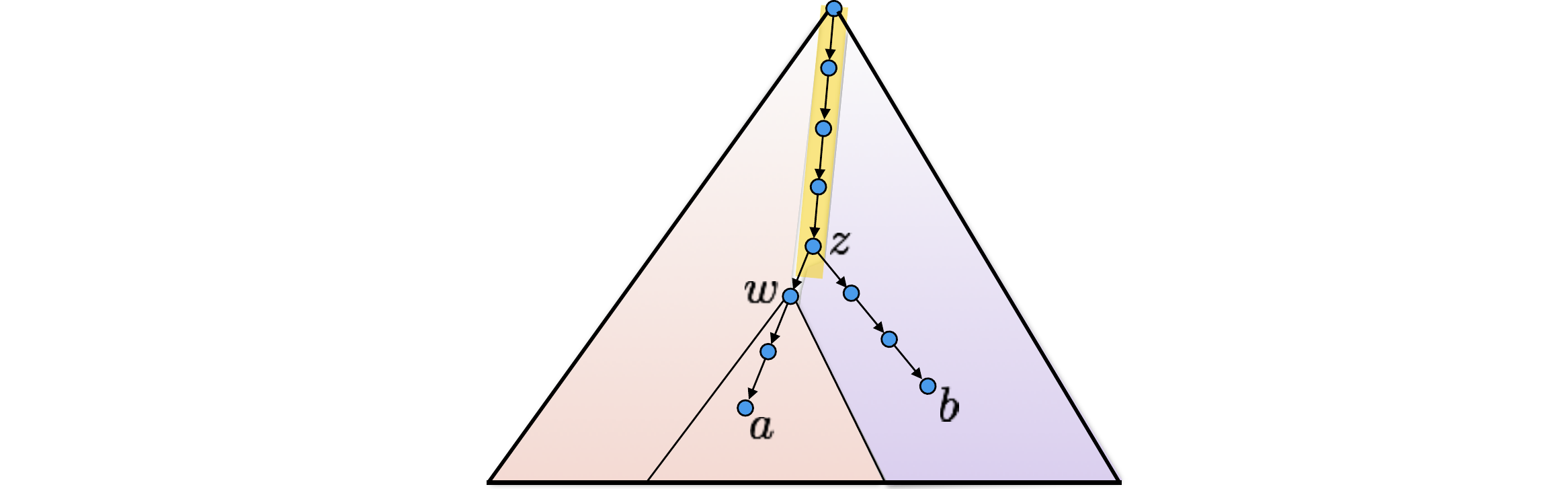}
\caption{ Depiction of vertices $a,b,z,w$ and sets $A$ (shown in orange) and $B$ (shown in purple).}
\label{figure:dfs}
\end{figure}

\subsection{A heavy path decomposition}

The heavy path decomposition of a tree was designed by Sleator and Tarjan~\cite{ST:83} in the context of dynamic trees. This
decomposition has been used in a variety of applications since then. Given any rooted tree $T$, this decomposition
splits $T$ into a set ${\cal P}$ of vertex disjoint paths with the property that any path from the root to a leaf
node in $T$ can be expressed as a concatenation of at most $\log n$ sub-paths of paths in ${\cal P}$. This decomposition is carried
out as follows. Starting from the root, we follow the path downward such that once we are at a node, say $v$, the
next node traversed is the child of $v$ in $T$ whose subtree is of maximum size, where the size of a subtree is the number
of nodes it contains. We terminate upon reaching
a leaf node. Let $P$ be the path obtained in this manner. If we remove $P$ from $T$, we are left with a collection of
subtrees each of size at most $n/2$. Each of these trees hang from $P$ through an edge in $T$. We carry out the
decomposition of these trees recursively.
The following lemma is immediate from the construction of a heavy path decomposition.

\begin{lemma}
For any vertex $v\in V$, the number of paths in $\cal P$ which start from either $v$ or an ancestor of $v$ in $T$
is at most $\log n$.
\label{lemma:heavy-path}
\end{lemma}

We now introduce the notion of ancestor path.

\begin{definition}
A path $Path(a_1,b_1)\in {\cal P}$ is said to be an ancestor path of
$Path(a_2,b_2)\in {\cal P}$, if $a_1$ is an ancestor of $a_2$ in $T$.
\end{definition}

In this paper, we describe the algorithm for computing SCCs of graph $G$ after any $k$ edge failures.
Vertex failures can be handled by simply splitting a vertex $v$ into an edge $(v_{in},v_{out})$,
where the incoming and outgoing edges of $v$ are directed to $v_{in}$ and from $v_{out}$, respectively.


\section{Computation of SCCs intersecting a given path}

Let $F$ be a set of at most $k$ failing edges, and
$X=(x_1,x_2,\ldots,x_t)$ be any path in $G$ from $x_1$ to $x_t$ which is intact in $G\setminus F$.
In this section, we present an algorithm that outputs in $O(2^k n \log n)$ time the
SCCs of $G\setminus F$ that intersect $X$.

For each $v\in V$, let $X^\In(v)$ be the vertex of $X$ of minimum index (if exists) that is
reachable from $v$  in $G\setminus F$. Similarly, let $X^\Out(v)$ be the vertex of $X$ of maximum index
(if exists) that has a path to $v$ in $G\setminus F$. (See Figure \ref{figure:in_out}).

\begin{figure}[!ht]
\centering
\includegraphics[scale=0.5, trim=20mm 80mm 40mm 73mm,clip]{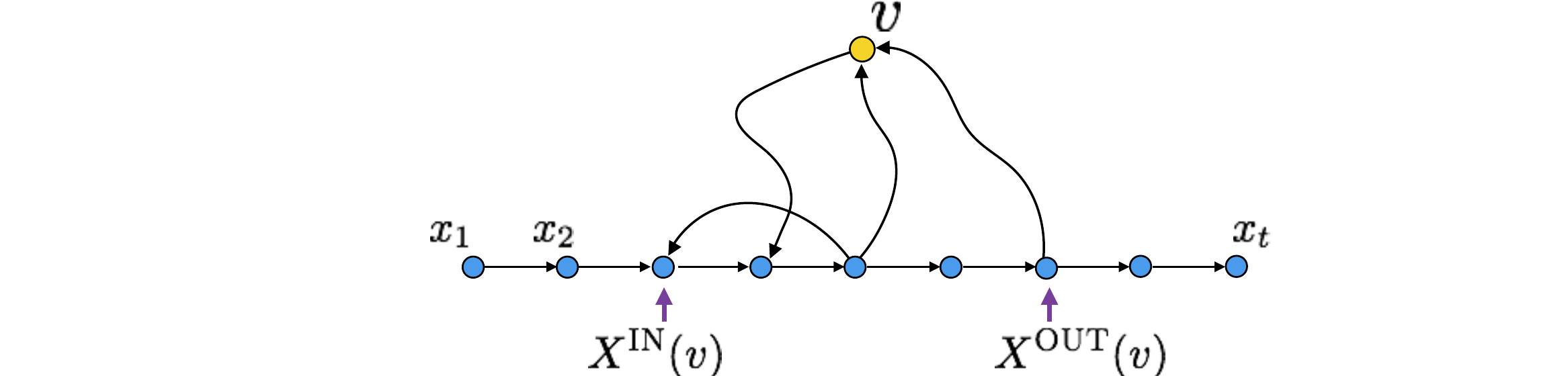}
\caption{ Depiction of $X^\In(v)$ and $X^\Out(v)$ for a vertex $v$ whose SCC intersects $X$.}
\label{figure:in_out}
\end{figure}

We start by proving certain conditions that must hold for a vertex if its SCC in $G\setminus F$ intersects $X$.

\begin{lemma}
For any vertex $w\in V$, the SCC that contains $w$ in $G\setminus F$ intersects $X$ if and only if the
following two conditions are satisfied.

(i) Both $X^\In(w)$ and $X^\Out(w)$ are defined, and

(ii) Either $X^\In(w)=X^\Out(w)$, or $X^\In(w)$ appears before $X^\Out(w)$ on $X$.
\label{lemma:SCC-intersect-X}
\end{lemma}
\begin{proof}
Consider any vertex $w\in V$. Let $S$ be the SCC in $G\setminus F$ that contains $w$ and assume $S$ intersects $X$.
Let $w_1$ and $w_2$ be the first and last vertices of $X$, respectively, that are in $S$.
Since $w$ and $w_1$ are in $S$ there is a path from $w$ to $w_1$ in $G\setminus F$.
Moreover, $w$ cannot reach a vertex that precedes $w_1$ in $X$ since such a vertex will be in $S$ as well and 
it will contradict the definition of $w_1$.  
Therefore, $w_1=X^\In(w)$. Similarly we can prove that $w_2=X^\Out(w)$.
Since $w_1$ and $w_2$ are defined to be the first and last vertices from $S$ on $X$, respectively, 
it follows that either $w_1=w_2$, or $w_1$ precedes $w_2$ on $X$.
Hence conditions (i) and (ii) are satisfied.

Now assume that conditions (i) and (ii) are true.
The definition of $X^\In(\cdot)$ and $X^\Out(\cdot)$ implies that there is a path from
$X^\Out(w)$ to $w$, and a path from $w$ to $X^\In(w)$. Also,  condition (ii) implies that there is a path from
$X^\In(w)$ to $X^\Out(w)$. Thus $w$, $X^\In(w)$ and $X^\Out(w)$ are in the same SCC and it intersects $X$.
\end{proof}

The following lemma states the condition under which any two vertices lie in the same SCC, given that their SCCs intersect $X$.

\begin{lemma}
Let $a,b$ be any two vertices in $V$ whose SCCs intersect $X$. Then $a$ and $b$ lie in the same SCC
if and only if $X^\In(a)=X^\In(b)$ and $X^\Out(a)=X^\Out(b)$.
\end{lemma}

\begin{proof}
In the proof of Lemma \ref{lemma:SCC-intersect-X}, we show that if SCC of $w$ intersects $X$,
then $X^\In(w)$ and $X^\Out(w)$ are precisely the first and last vertices on $X$ that lie in the SCC of $w$.
Since SCCs forms a partition of $V$, vertices $a$ and $b$ will lie in the same SCC if and only if
$X^\In(a)=X^\In(b)$ and $X^\Out(a)=X^\Out(b)$.
\end{proof}

It follows from the above two lemmas that in order to compute the SCCs in $G\setminus F$ that intersect with $X$,
it suffices to compute $X^\In(\cdot)$ and $X^\Out(\cdot)$ for all vertices in $V$.
It suffices to focus on computation of $X^\Out(\cdot)$ for all the vertices of $V$, since $X^\In(\cdot)$ can be computed
in an analogous manner by just looking at graph $G^R$.
One trivial approach to achieve this goal is to compute the set $V_i$ consisting of
all vertices reachable from each $x_i$ by performing a BFS or DFS traversal of graph ${\cal G}(x_i)\setminus F$.
Using this straightforward approach it takes $O(2^knt)$ time to complete the task of computing
$X^\Out(v)$ for every $v\in V$, while our target is to do so in $O(2^kn \log n)$ time.

Observe the nested structure underlying $V_i$'s, that is, $V_1\supseteq V_2\supseteq \cdots \supseteq V_t$.
Consider any vertex $x_\ell, 1<\ell <t$. The nested structure implies for every $v\in V_\ell$ that $X^\Out(v)$
must be on the portion $(x_\ell,\ldots,x_t)$ of $X$. Similarly, it implies
for every $v\in V_1\setminus V_\ell$ that $X^\Out(v)$ must be on the portion $(x_1,\ldots,x_{\ell-1})$ of $X$.
This suggests a  divide and conquer approach to efficiently compute $X^\Out(\cdot)$.
We first compute the sets $V_1$ and $V_t$ in $O(2^k n)$ time each.
For each $v\in V\setminus V_1$, we assign NULL to $X^\Out(v)$ as it is not reachable from any vertex on $X$;
and for each $v\in V_t$ we set $X^\Out(v)$ to $x_t$.
For vertices in set $V_1\setminus V_t$, $X^\Out(\cdot)$ is computed by calling the function
Binary-Search($1,t-1,V_1\setminus V_t$). See Algorithm \ref{Binary-Search}.

\begin{algorithm}
\uIf {$(i=j)$}{\lForEach{$v\in A$}{$X^{\Out}(v)=x_i$}}
\Else {
$mid \gets \lceil(i+j)/2\rceil$\;
$B \gets$ Reach($x_{mid},A$)\tcc*[r]{vertices in $A$ reachable from $x_{mid}$}
Binary-Search($i,mid\text{-}1,A\text{$\setminus$} B$)\;
Binary-Search($mid,j,B$)\;
}
\caption{Binary-Search($i,j,A$)}
\label{Binary-Search}
\end{algorithm}

In order to explain the function Binary-Search, we first state an assertion that
holds true for each recursive call of the function Binary-Search. We prove this assertion in the next subsection.
\begin{description}
\item[Assertion 1:]
If Binary-Search($i,j,A$) is called, then $A$ is precisely the set of those vertices
$v\in V$ whose $X^\Out(v)$ lies on the path $(x_i,x_{i+1},\ldots,x_j)$.
\end{description}

We now explain the execution of function Binary-Search($i,j,A$).
If $i=j$, then we assign $x_i$ to $X^{\Out}(v)$ for each $v\in A$ as justified by Assertion 1.
Let us consider the case when $i\neq j$.
In this case we first compute the index $mid = \lceil(i+j)/2\rceil$.
Next we compute the set $B$ consisting of all the vertices in $A$ that are reachable from $x_{mid}$.
This set is computed using the function Reach($x_{mid},A$) which is explained later in Subsection \ref{function_Reach}.
As follows from Assertion 1, $X^\Out(v)$ for each vertex $v\in A$ must belong to path $(x_i,\ldots,x_j)$.
Thus, $X^{\Out}(v)$ for all $v\in B$ must lie on path $(x_{mid},\ldots,x_j)$, and
$X^{\Out}(v)$ for all $v\in A\setminus B$ must lie on path $(x_i,\ldots,x_{mid\text{-}1})$.
So for computing $X^\Out(\cdot)$ for vertices in $A\setminus B$ and $B$, we invoke the functions
Binary-Search($i,mid\text{-}1,A\text{$\setminus$} B$) and Binary-Search($mid,j,B$), respectively.


\subsection{Proof of correctness of algorithm}

In this section we prove that Assertion 1 holds for each call of the Binary-Search function. We also show how this
assertion implies that $X^\Out(v)$ is correctly computed  for every $v\in V$.

Let us first see how Assertion 1 implies the correctness of our algorithm.
It follows from the description of the algorithm that for each $i, (1\leq i\leq t-1)$, the function
Binary-Search($i,i,A$) is invoked for some $A\subseteq V$. Assertion 1 implies that $A$
must be the set of all those vertices $v\in V$ such that $X^\Out(v)=x_i$. As can be seen,
the algorithm in this case correctly sets $X^\Out(v)$ to $x_i$ for each $v\in A$.

We now show that Assertion 1 holds true in each call of the function Binary-Search.
It is easy to see that Assertion 1 holds true for the first call
Binary-Search($1,t-1,V_1\setminus V_t$). Consider any intermediate recursive
call Binary-Search($i,j,A$), where $i\neq j$. It suffices to show that if Assertion 1
holds true for this call, then it also holds true for the two recursive calls that it invokes.
Thus let us assume $A$ is the set of those vertices $v\in V$ whose $X^\Out(v)$ lies on the path $(x_i,x_{i+1},\ldots,x_j)$.
Recall that we compute index $mid$ lying between $i$ and $j$, and find the set $B$ consisting of
all those vertices in $A$ that are reachable from $x_{mid}$.
From the nested structure of the sets $V_i,V_{i+1},\ldots, V_j$, it follows that $X^{\Out}(v)$ for all $v\in B$
must lie on path $(x_{mid},\ldots,x_j)$, and
$X^{\Out}(v)$ for all $v\in A\setminus B$ must lie on path $(x_i,\ldots,x_{mid\text{-}1})$.
That is, $B$ is precisely the set of those vertices whose $X^\Out(v)$ lies on the path $(x_{mid},\ldots,x_j)$,
and $A\setminus B$ is precisely the set of those vertices whose $X^\Out(v)$ lies on the path $(x_i,\ldots,x_{mid\text{-}1})$.
Thus Assertion 1 holds true for the recursive calls  Binary-Search($i,mid\text{-}1,A\text{$\setminus$} B$) and
Binary-Search($mid,j,B$) as well.


\subsection{Implementation of function Reach}
\label{function_Reach}

The main challenge left now is to find an efficient implementation of the function Reach 
which has to compute the  vertices of its input set $A$ that are reachable from a given vertex $x \in X$
in  $G\setminus F$. The function Reach can be easily implemented  by a standard graph traversal initiated from $x$
in the graph  ${\cal G}(x)\setminus F$ (recall that  ${\cal G}(x)$ is a $\kftrs$ of $x$ in $G$).  
This, however,  will take $O(2^k n)$ time which is not good enough for our purpose, as the total running time 
of Binary-Search in this case will become $O(|X|2^k n)$. 
Our aim is to implement the function Reach in $O(2^k|A|)$ time. 
In general, for an arbitrary set $A$ this might not be possible.
This is because $A$ might contain a vertex that is reachable from $x$ via a single path whose vertices are 
not in $A$, therefore, the algorithm must  explore edges incident to  vertices that are not in   $A$ as well.
However, the following lemma, that exploits Assertion 1, suggests that in our case as the call to Reach is done 
while running the function Binary-Search we can restrict ourselves to the set $A$ only.

\begin{lemma}
If \normalfont{Binary-Search}$(i,j,A)$ is called and $\ell\in[i,j]$, then 
for each path $P$ from $x_\ell$ to a vertex $z\in A$ in graph in $G\setminus F$, 
all the vertices of $P$ must be in the set $A$.
\label{lemma:set_A}
\end{lemma}

\begin{proof}
Assertion 1 implies that $A$ is precisely the set of those vertices in $V$ which are
reachable from $x_i$ but not reachable from $x_{j+1}$ in $G\setminus F$.
Consider any vertex $y\in P$. Observe that $y$ is reachable from $x_i$ by the path $X[x_i,x_\ell]$::$P[x_\ell,y]$.
Moreover, $y$ is not reachable from $x_{j+1}$, because otherwise $z$ will also be reachable from $x_{j+1}$,
which is not possible since $z\in A$. Thus vertex $y$ lies in the set $A$.
\end{proof}

Lemma \ref{lemma:set_A} and Lemma \ref{lemma:kftrs-set} imply that in order to find
the vertices in $A$ that are reachable from $x_{mid}$, it suffices to
do traversal from $x_{mid}$ in the graph $G_A$, the induced subgraph of $A$ in   ${\cal G}(x)\setminus F$, 
that has $O(2^k |A|)$ edges. 
Therefore, based on the above discussion, Algorithm \ref{Reach} given below, is an implementation
of function Reach that takes $O(2^k |A|)$ time.

\begin{algorithm}
$H \gets {\cal G}(x_{mid})\setminus F$\;
$G_A \gets (A,\emptyset)$\tcc*[r]{an empty graph}
\ForEach {$v\in A$}{
\ForEach {$(y,v)\in \inedges(v,H)$}{
\lIf{$y\in A$}{$E(G_A)=E(G_A)\cup (y,v)$}
}
}
$B \gets$ Vertices reachable from $x_{mid}$ obtained by a BFS or DFS traversal of graph $G_A$\;
Return $B$\;
\caption{Reach($x_{mid},A$)}
\label{Reach}
\end{algorithm}

The following lemma gives the analysis of running time of Binary-Search($1,t-1,V_1\setminus V_t$).

\begin{lemma}
The  total running time of \normalfont{Binary-Search}$(1,t-1,V_1\setminus V_t)$ is $O(2^{k} n \log n)$.
\label{lemma:time-binary-search}
\end{lemma}
\begin{proof}
The time complexity of Binary-Search($1,t-1,V_1\setminus V_t$) is dominated by the total time taken by
all invocation of function Reach.
Let us consider the recursion tree associated with Binary-Search($1,t-1,V_1\setminus V_t$). It can be seen that
this tree will be of height $O(\log n)$. 
In each call of the Binary-Search, the input set $A$ is partitioned  into two disjoint sets.
As a result, the input sets associated with all recursive calls at any level $j$ in
the recursion tree form a disjoint partition of $V_1\setminus V_t$.
Since the time taken by Reach is $O(2^{k}|A|)$, so the total time taken by all invocations
of Reach at any level $j$ is $O(2^{k}|V_1\setminus V_t|)$.
As there are at most $\log n$ levels in the recursion tree, the
total time taken by Binary-Search($1,t-1,V_1\setminus V_t$) is $O(2^{k} n \log n)$.
\end{proof}

We conclude with the following theorem.

\begin{theorem}
\label{theorem:specific_case}
Let $F$ be any set of at most $k$ failed edges, and $X=\{x_1,x_2,\ldots,x_t\}$ be any path in $G\setminus F$.
If we have prestored the graphs ${\cal G}(x)$ and ${\cal G}^R(x)$ for each $x\in X$, then we can compute all the 
SCCs of $G\setminus F$ which intersect with $X$ in $O(2^k n\log n)$ time.
\end{theorem}


\section{Main Algorithm}
\label{section:main-algo}

In the previous section we showed that given any path $P$, we can compute all the SCCs intersecting $P$ efficiently,
if $P$ is intact in $G\setminus F$.
In the case that $P$ contains $\ell$ failed edges from $F$ then $P$ is decomposed into $\ell+1$ paths, and
we can apply Theorem \ref{theorem:specific_case} to each of these paths separately to get the following theorem:

\begin{theorem}
Let $P$ be any given path in $G$. Then there exists an $O(2^k n|P|)$ size data structure
that for any arbitrary set $F$ of at most $k$ edges computes the SCCs of $G\setminus F$
that intersect the path $P$ in $O((\ell+1) 2^k n \log n)$ time, where 
$\ell~$ ($\ell\leq k$) is the number of edges in $F$ that lie on $P$.  
\label{theorem:split-path}
\end{theorem}

Now in order to use Theorem \ref{theorem:split-path} to design a fault tolerant
algorithm for SCCs, we need to find a family of paths, say $\cal P$,
such that for any $F$, each SCC of $G\setminus F$ intersects at least one path in $\cal P$.
As described in the Subsection \ref{subsection:overview}, a heavy path decomposition of DFS tree $T$
serves as a good choice for $\cal P$.
Choosing $T$ as a DFS tree helps us because of the following reason: let $P$ be any root-to-leaf path,
and suppose we have already computed the SCCs in $G\setminus F$ intersecting $P$. Then each of the remaining SCCs must be
contained in some subtree hanging from path $P$.
The following lemma formally states this fact.

\begin{lemma}
Let $F$ be any set of failed edges, and $Path(a,b)$ be any path in $\cal P$.
Let $S$ be any SCC in $G\setminus F$ that intersects $Path(a,b)$ but does not
intersect any path that is an ancestor path of $Path(a,b)$ in $\cal P$.
Then all the vertices of $S$ must lie in the subtree $T(a)$.
\label{lemma:subtree}
\end{lemma}
\begin{proof}
Consider a vertex $u$ on $Path(a,b)$ whose SCC $S_u$ in $G\setminus F$
is not completely contained in the subtree $T(a)$.
We show that $S_u$ must contain an ancestor of $a$ in $T$, thereby
proving that it intersects an ancestor-path of $Path(a,b)$ in $\cal P$.
Let $v$ be any vertex in $S_u$ that is not in the subtree $T(a)$.
Let $P_{u,v}$ and $P_{v,u}$ be paths from $u$ to $v$
and from $v$ to $u$, respectively,  in $G\setminus F$.
From Lemma \ref{lemma:DFS-property} it follows that either $P_{u,v}$ or $P_{v,u}$
must pass through a common ancestor of $u$ and $v$ in $T$. Let this ancestor be $z$.
Notice also that since $P_{u,v}$ and $P_{v,u}$ form a cycle all their vertices are in $S_u$. Therefore,
$u$ and $z$ are in the same SCC in $G\setminus F$.
Moreover, since $v\notin T(a)$ and $u\in T(a)$, their common ancestor $z$ in $T$ is an ancestor of $a$.
Since $z\in S_u$ and it is an ancestor of $a$ in $T$, the lemma follows.
\end{proof}

Lemma \ref{lemma:subtree} suggests that if we process the paths from ${\cal P}$ in
the non-decreasing order of their depths, then in order to compute the SCCs intersecting a path $Path(a,b)\in \cal P$,
it suffices to focus on the subgraph induced by the vertices in $T(a)$ only. This is because the SCCs intersecting $Path(a,b)$
that do not completely lie in $T(a)$ would have already been computed during the processing of some ancestor
path of $Path(a,b)$.

We preprocess the graph $G$ as follows. We first compute a heavy path decomposition $\cal P$ of DFS tree $T$.
Next for each path $Path(a,b)\in \cal P$, we use Theorem \ref{theorem:split-path} to construct the data structure
for path $Path(a,b)$ and the subgraph of $G$ induced by vertices in $T(a)$. We use the notation
${\cal D}_{a,b}$ to denote this data structure.
Our algorithm for reporting SCCs in $G\setminus F$ will use the collection of these data structures associated with
the paths in $\cal P$ as follows.

Let $\cal C$ denote the collection of SCCs in $G\setminus F$ initialized to $\emptyset$.
We process the paths from ${\cal P}$ in non-decreasing order of their depths. Let $P(a,b)$ be any path
in ${\cal P}$ and let $A$ be the set of vertices belonging to $T(a)$.
We use the data structure ${\cal D}_{a,b}$ to compute SCCs of $G(A)\setminus F$ intersecting $P(a,b)$.
Let these be $S_1,\ldots,S_t$. Note that some of these SCCs might be a part of some
bigger SCC computed earlier. We can detect it by keeping a set $W$ of all vertices for which we have computed
their SCCs. So if $S_i\subseteq W$, then we can discard $S_i$, else we add $S_i$ to collection $\cal C$.
Algorithm \ref{main-algo} gives the complete pseudocode of this algorithm.

\begin{algorithm}
$\cal C \gets \emptyset$\tcc*[r]{Collection of SCCs}
$W \gets \emptyset$\tcc*[r]{A subset of $V$ whose SCC have been computed}
${\cal P}\gets$ A heavy-path decomposition of $T$, where paths are sorted in the non-decreasing order of their depths\;
\ForEach{$Path(a,b)\in {\cal P}$}{
  $A\gets$ Vertices lying in the subtree $T(a)$\;
  $(S_1,\ldots,S_t) \gets$ SCCs intersecting $Path(a,b)$ in $G(A)\setminus F$ computed using ${\cal D}_{a,b}$\;
  \ForEach{$i\in [1,t]$}{
    \lIf{$(S_i\nsubseteq W)$}{Add $S_i$ to collection $\cal C$ and set $W=W\cup S_i$}
  }
}
Return $~\cal C$\;
\caption{Compute SCC($G,F$)}
\label{main-algo}
\end{algorithm}

Note that, in the above explanation, we only used the fact that $T$ is a DFS tree, and $\cal P$ could have been any path decomposition of $T$.
We now show how the fact that $\cal P$ is a heavy-path decomposition is crucial for the efficiency of our algorithm.
Consider any vertex $v\in T$. The number of times $v$ is processed in Algorithm \ref{main-algo} is equal to the number of
paths in $\cal P$ that start from either $v$ or an ancestor of $v$. For this number to be small for each $v$, we choose
$\cal P$ to be a heavy path decomposition of $T$.
On applying Theorem \ref{theorem:split-path}, this immediately
gives that the total time taken by Algorithm \ref{main-algo} is $O(k2^kn\log^2n)$.
In the next subsection, we do a more careful analysis and show that this bound can be improved to $O(2^kn\log^2n)$.


\subsection{Analysis of time complexity of Algorithm \ref{main-algo}}
\label{section:improved-time-complexity}

For any path $Path(a,b)\in \cal P$ and any set $F$ of failing edges, let $\ell(a,b)$ denote the number
of edges of $F$ that lie on $Path(a,b)$. It follows from Theorem \ref{theorem:split-path} that
the time spent in processing $Path(a,b)$ by Algorithm \ref{main-algo} is
$O\big((\ell(a,b)+1)\times 2^k|T(a)| \times \log n\big)$. 
Hence the time complexity of Algorithm \ref{main-algo} is of the order of 

$$\sum_{Path(a,b)\in{\cal P}} (\ell(a,b)+1)\times 2^k|T(a)| \times \log n $$

In order to calculate this we define a notation $\alpha(v,Path(a,b))$ as $\ell(a,b)+1$ if  $v\in T(a)$,
and $0$ otherwise, for each $v\in V$ and $Path(a,b)\in \cal P$.
So the time complexity of Algorithm \ref{main-algo} becomes

$$~2^k \log n \times \Big(\sum_{Path(a,b)\in{\cal P}} (\ell(a,b)+1)\times |T(a)|\Big)$$
$$=2^k \log n \times \Big(\sum_{Path(a,b)\in{\cal P}}~ \sum_{v\in V} \alpha(v,Path(a,b)) \Big)$$
$$=2^k \log n \times \Big(\sum_{v\in V} ~\sum_{Path(a,b)\in{\cal P}} \alpha(v,Path(a,b)) \Big)$$

Observe that for any vertex $v$ and $Path(a,b)\in \cal P$, $\alpha(v,Path(a,b))$ is equal to 
$\ell(a,b)+1$ if $a$ is either $v$ or an ancestor of $v$, otherwise it is zero.
Consider any vertex $v\in V$.
We now show that $\sum_{Path(a,b)\in{\cal P}} \alpha(v,Path(a,b))$ is at most $k+\log n$.
Let $P_v$ denote the set of those paths in $\cal P$ which starts from either $v$ or an ancestor of $v$.
Then  $\sum_{Path(a,b)\in{\cal P}} \alpha(v,Path(a,b)) =  \sum_{Path(a,b)\in{P_v}} \ell(a,b)+1$.
Note that $\sum_{Path(a,b)\in{P_v}} \ell(a,b)$ is at most $k$, and Lemma \ref{lemma:heavy-path} implies that the 
number of paths in $P_v$ is at most $\log n$. 
This shows that $\sum_{Path(a,b)\in{\cal P}} \alpha(v,Path(a,b))$ is at most $k+\log n$
which is $O(\log n)$, since $k\leq \log n$.

Hence the time complexity of Algorithm \ref{main-algo} becomes $O(2^k n\log^2 n)$.
We thus conclude with the following theorem.

\begin{theorem}
For any $n$-vertex directed graph $G$, there exists an $O(2^k n^2)$ size
data structure that, given any set $F$ of at most $k$ failing edges, can report 
all the SCCs of $G\setminus F$ in $O(2^k n\log^2 n)$ time.
\label{theorem:main-theorem}
\end{theorem}


\section{Extension to handle insertion as well as deletion of edges}
\label{section:updates}

In this section we extend our algorithm to incorporate insertion as well as deletion of edges. 
That is,  we describe an algorithm for reporting SCCs of a directed graph $G$ when there are 
at most $k$ edge insertions and at most $k$ edge deletions.

Let $\cal D$ denote the $O(2^k n^2)$ size data structure, described in Section \ref{section:main-algo},
for handling $k$ failures. In addition to $\cal D$, we store the two $\kftrs$: ${\cal G}^{}(v)$ and 
${\cal G}^{R}(v)$ for each vertex $v$ in $G$. Thus the space used remains the same, i.e. $O(2^k n^2)$.
Now let $U=(X,Y)$ be the ordered pair of $k$ updates, with $X$ being the set of failing edges and $Y$ being 
the set of newly inserted edges. Also let $|X|\le k$ and $|Y|\le k$.

\begin{algorithm}[!ht]
$\cal C \gets$ SCCs of graph $G\setminus X$ computed using data structure $\cal D$\;
$S \gets$ Subset of $V$ consisting of endpoints of edges in $Y$\;
$H \gets \bigcup_{v\in S} \big({\cal G}^{}(v) + {\cal G}^{R}(v) + Y\big)$\;
Compute SCCs of graph $H\setminus X$ using any standard static algorithm\;
\ForEach{$v\in S$}
{
Merge all the smaller SCCs of $\cal C$ which are contained in $SCC_{H\setminus X}(v)$ into a single SCC\;
}
\caption{Find-SCCs($U=(X,Y)$)}
\label{algo:k-updates}
\end{algorithm}

Our first step is to compute the collection $\cal C$, consisting of SCCs of graph $G\setminus X$.
This can be easily done in $O(2^k n \log^2 n)$ time using the data structure $\cal D$.
Now on addition of set $Y$, some of the SCCs in $\cal C$ may get merged into bigger SCCs.
Let $S$ be the subset of $V$ consisting of endpoints of edges in $Y$.
Note that if the SCC of a vertex gets altered on addition of $Y$, then its new SCC 
must contain at least one edge from $Y$, and thus also a vertex from set $S$.
Therefore, in order to compute SCCs of $G+U$, it suffices to recompute only the SCCs of
vertices lying in the set $S$.

\begin{lemma}
Let $H$ be a graph consisting of edge set $Y$, and the $\kftrs$ ${\cal G}^{}(v)$
and ${\cal G}^{R}(v)$, for each $v\in S$. Then $SCC_{H\setminus X}(v)=SCC_{G+U}(v)$, for each $v\in S$.
\label{lemma:k-updates}
\end{lemma}

\begin{proof}
Consider a vertex $v\in S$. Since $H\setminus X\subseteq G+U$, $SCC_{H\setminus X}(v)\subseteq SCC_{G+U}(v)$.
We show that  $SCC_{H\setminus X}(v)$ is indeed equal to $SCC_{G+U}(v)$.

Let $w$ be any vertex reachable from $v$ in $G+U$, by a path, say $P$.
Our aim is to show that $w$ is reachable from $v$ in $H\setminus X$ as well.
Notice that we can write $P$ as $(P_1\concat e_1 \concat P_2 \concat e_2 \cdots e_{\ell-1} \concat P_\ell)$,
where $e_1,\ldots,e_{\ell-1}$ are edges in $Y\cap P$ and $P_1,\ldots,P_\ell$ are segments of $P$
obtained after removal of edges of set $Y$. Thus $P_1,\ldots,P_\ell$ lie in $G\setminus X$.
For $i=1$ to $\ell$, let $a_i$ and $b_i$ be respectively the first and last vertices of path $P_i$.
Since $a_1=v$ and $a_2,\ldots,a_\ell\in S$, the $\kftrs$ of all the vertices $a_1$ to $a_\ell$ is contained in $H$. 
Thus for $i=1$ to $\ell$, vertex $b_i$ must be reachable from $a_i$ by some path, say $Q_i$, in graph $H\setminus X$.
Hence $Q = (Q_1\concat e_1 \concat Q_2 \cdots e_{\ell-1} \concat Q_\ell)$ is a path from $a_1=v$ to $b_\ell=w$
in graph $H\setminus X$. 

In a similar manner we can show that if a vertex $w'$ has a path to $v$ in  graph $G+U$, then 
$w'$ will also have path to $v$ in graph $H\setminus X$. Thus $SCC_{H\setminus X}(v)$ must be equal to $SCC_{G+U}(v)$.
\end{proof}

So we compute the auxiliary graph $H$ as described in Lemma \ref{lemma:k-updates}.
Note that $H$ contains only $O(k 2^k n)$ edges. 
Next we compute the SCCs of graph $H\setminus X$ using any standard algorithm \cite{Cormen} that runs in time
which is linear in terms of the number of edges and vertices. 
This algorithm will take $O(2^k n \log n)$ time, since $k$ is at most $\log n$.
Finally, for each $v\in S$, we check if the $SCC_{H\setminus X}(v)$ has broken into smaller SCCs in $\cal C$,
if so, then we merge all of them into a single SCC. We can accomplish this entire task in a total $O(nk)$ time only. 
This completes the description of our algorithm. For the pseudocode see Algorithm \ref{algo:k-updates}.

We conclude with the following theorem.

\begin{theorem}
For any $n$-vertex directed graph $G$, there exists an $O(2^k n^2)$ size data structure that, 
given any set $U$ of at most $k$ edge insertions and at most $k$ edge deletions,
can report the SCCs of graph $G+U$ in $O(2^k n\log^2 n)$ time.
\label{theorem:k-updates}
\end{theorem}


\bibliography{ref}

\end{document}